\documentclass[11pt]{article}

\usepackage[margin=1in]{geometry}
\usepackage{microtype}

\usepackage{amsmath,amssymb,amsfonts,amsthm}
\usepackage{booktabs}
\usepackage{multirow}
\usepackage{array}

\usepackage{graphicx}
\usepackage{float}      
\usepackage{placeins}   
\usepackage{enumitem}

\usepackage{algorithm}
\usepackage{algpseudocode}

\usepackage{tikz}
\usetikzlibrary{arrows.meta,positioning,calc,fit,shapes.geometric}

\usepackage{hyperref}
\usepackage{url}
\hypersetup{
  colorlinks=true,
  linkcolor=blue,
  citecolor=blue,
  urlcolor=blue
}

\newtheorem{definition}{Definition}
\newtheorem{theorem}{Theorem}
\newtheorem{proposition}{Proposition}

\setlist{itemsep=2pt, topsep=4pt, leftmargin=*}
\title{\textbf{Contract2Plan: Verified Contract-Grounded Retrieval-Augmented Optimization\\
for BOM-Aware Procurement and Multi-Echelon Inventory Planning}}

\author{
  Sahil Agarwal%
  \thanks{Correspondence: \texttt{sahil@foundincache.com}}%
  \thanks{The author acknowledges the use of AI language assistants for stylistic editing and reference checking. All scientific content, modeling choices, analysis, and conclusions are the author’s own.}\\
  Independent Researcher \textbar\ Found in Cache
}

\date{\today}

\begin{document}
\maketitle

\begin{abstract}
Procurement and inventory planning in real supply chains is governed not only by demand forecasts and bills of materials (BOMs), but also by operational terms embedded in contracts, tenders, and supplier documents: minimum order quantities (MOQs), lead-time commitments, price breaks, allocation caps, penalties, delivery terms, and substitution approvals. Large language models (LLMs) can extract structured constraints from such text, yet extraction-only or LLM-only decision pipelines are brittle. Missed clauses, unit errors, and unresolved contradictions can silently create infeasible plans or contract violations, and BOM coupling amplifies small upstream errors.

We introduce \textbf{Contract2Plan}, a verified GenAI-to-optimizer pipeline that (i) retrieves relevant clauses from multi-source documents with provenance, (ii) extracts a typed constraint schema with evidence spans and confidence, (iii) compiles extracted constraints into a BOM-aware mixed-integer planning model, and (iv) applies a solver-driven \textbf{verifier} that enforces grounding and eligibility, detects contradictions and missing fields, checks feasibility (via IIS or slack-minimization diagnosis), and triggers an iterative \textbf{repair loop}. Crucially, we strengthen the core claim with a compliance section: for monotone feasibility constraints (MOQ, lead time, capacity, cadence), conservative repair yields contract-safe feasibility under explicit assumptions; for non-monotone or exception-heavy clauses, the system abstains and escalates to humans rather than guessing.

To quantify tail risk in a risk-illustrative setting, we report a self-contained synthetic micro-benchmark (500 instances; horizon $T=5$) computed by exact enumeration over $9^5=59{,}049$ order schedules per instance under a realistic execution model (suppliers uplift sub-MOQ orders; late arrivals trigger emergency purchases). Extraction-only planning exhibits heavy-tailed regret relative to verified planning: mean \$142.33 (5.40\% of mean optimal cost; 95\% bootstrap CI \$[113.67,171.07]), 90th percentile \$587.74, 99th percentile \$1569.61, maximum \$2242.22, and planned MOQ-violation incidence 16.6\% (95\% bootstrap CI [13.4\%, 20.0\%]).
\end{abstract}


\section{Introduction}
Procurement terms are operational constraints. A planner may have a demand forecast, a detailed BOM, and a sound replenishment heuristic, yet still fail in execution if a contract addendum tightens an MOQ, if a seasonal clause extends lead times, if an allocation cap throttles supply, or if a discount tier is applied without meeting eligibility thresholds. In practice these terms are scattered across heterogeneous sources: master supply agreements, exhibits, addenda, tender attachments, catalogs, emails, and internal policy documents. Translating this text into structured, auditable planning inputs remains labor-intensive and error-prone.

LLMs are attractive for accelerating translation from text to structured fields, and recent work explores LLM agents for inventory management \cite{quan2024invagent}. However, using LLM outputs directly for decisions (or trusting extracted constraints without verification) is unsafe in high-coupling settings (BOM explosion, multi-echelon flows). Common failure modes include:
\begin{itemize}
\item \textbf{Omission:} relevant clauses are not retrieved or not extracted.
\item \textbf{Hallucination and unit errors:} a model invents a value or mis-parses units (days vs weeks; currency; tier thresholds).
\item \textbf{Inconsistency:} multiple documents disagree (e.g., MOQ 100 in the master agreement, MOQ 150 in an addendum).
\item \textbf{Condition or scope miss:} applicability is wrong (effective dates, regions, product scope).
\end{itemize}

These errors are not cosmetic. Small mistakes can flip feasibility, induce late arrivals that trigger emergency buys, or silently violate contractual terms. Because procurement and production are coupled through BOMs and lead times, a single upstream misread clause can cascade into multi-period service failures.

\paragraph{Illustrative failure vignette (synthetic but realistic).}
Consider a manufacturer with a two-level BOM where a single constrained component $p$ gates assembly of a high-volume finished good $f$.
The master agreement states an MOQ of 100 units and a 6-week lead time for $p$.
An addendum (effective May 1) raises MOQ to 150 for deliveries to site MX-01 and introduces a peak-season lead time of 10 weeks for orders placed in Aug--Oct, while a pricing exhibit offers a discounted tier at 150 units per PO line.
A text-to-JSON extractor that retrieves only the master agreement (or mis-scopes the addendum) may plan an order of 120 units in July and assume 6 weeks lead time.
In execution, the supplier uplifts the order to 150 (changing cash flow and inventory) and, if the order is placed in peak season, ships later than planned, forcing emergency buys or production shortfalls.
The resulting failure is not merely higher cost: it creates a compliance exposure (tier eligibility and MOQ terms), a service impact (late arrivals), and a compounded planning error because BOM coupling propagates the component mistake to multiple periods and products.
This paper treats such failures as the default risk profile of contract-grounded planning and designs verification as a first-class gate before decisions are emitted.

\textbf{Core idea.} Contract2Plan separates \emph{constraint extraction} from \emph{decision optimization} and inserts a verifier that treats extracted constraints as hypotheses: a constraint can influence planning only after it is grounded to evidence, validated for internal consistency, and shown to be solver-feasible (or conservatively repaired where safe).

\textbf{Contributions.}
\begin{itemize}
\item \textbf{Task definition:} contract-grounded, BOM-aware procurement and planning with auditable provenance.
\item \textbf{Method:} a Verified RAG-to-Optimization loop (retrieval $\rightarrow$ extraction $\rightarrow$ compilation $\rightarrow$ solver verification $\rightarrow$ repair $\rightarrow$ explanation).
\item \textbf{Compliance guarantees:} a theorem-level statement of which contractual constraint classes are safely handled by conservative repair and which require human confirmation.
\item \textbf{Auditability:} decision cards linking each material constraint and recommendation back to evidence spans.
\item \textbf{Computed risk quantification:} a reproducible micro-benchmark computed by exact enumeration demonstrating heavy-tailed economic and compliance risk from MOQ/lead-time extraction errors.
\end{itemize}

\section{Related Work and Positioning}
\textbf{LLMs for inventory and supply chain decision systems.} InvAgent proposes a large language model based multi-agent system for inventory management \cite{quan2024invagent}. Contract2Plan differs in a foundational way: it does not ask an LLM to be the final decision-maker. Instead, it uses LLMs for evidence retrieval and schema-constrained extraction, and insists that feasibility and compliance are enforced by explicit schemas and optimization solvers.

\textbf{LLMs for optimization modeling and tool use.} OptiBench emphasizes that realistic optimization problem solving requires solver calls and checkable outputs \cite{yang2024optibench}. Contract2Plan follows the tool-use principle but targets a different failure surface: the correctness, scope, and precedence of contract-derived constraints that define the optimization model.

\textbf{MILP learning and benchmarks.} A foundation-model perspective for MILP has motivated large-scale instance generation via MILP-Evolve \cite{li2024milp}. Contract2Plan is complementary: we focus on reading and verifying constraints from procurement text rather than learning to solve MILPs.

\textbf{Contract understanding and clause retrieval.} CUAD \cite{hendrycks2021cuad} and ContractNLI \cite{koreeda2021contractnli} show that contract clause identification and evidence grounding are challenging even on expert-annotated datasets. ACORD provides an expert-annotated retrieval benchmark for contract clause retrieval \cite{wang2025acord}. Contract2Plan uses similar building blocks (retrieval, evidence spans) but couples them to optimization compilation and solver verification.

\textbf{Inventory control and hybrid methods.} Deep reinforcement learning has been surveyed for inventory control, highlighting modeling and evaluation challenges \cite{boute2022drlroadmap}. Hybrid methods combining DRL with stochastic programming have also been explored \cite{stranieri2024drlstoch}. Contract2Plan is compatible with these approaches, but addresses a different bottleneck: transforming contract text into verified constraints with auditable compliance.

\section{Background and Problem Setting}
\subsection{BOM-aware planning}
A bill of materials (BOM) is a directed acyclic graph linking finished goods to subassemblies and components. Let $a_{p,f}$ denote the units of component $p$ required per unit of finished good $f$. Shared components couple multiple finished goods, and small errors in component-level constraints (MOQ, lead time, capacity) can propagate into service failures.

BOMs commonly include alternates and substitutes subject to approved vendor lists (AVL) and engineering or quality approvals. Contracts can further constrain substitution (forbidden unless approved, approval required for certain lots, or restricted for regulated components). These policies must be treated as hard constraints.

\subsection{Multi-echelon inventory planning}
Multi-echelon networks hold inventory at multiple locations (suppliers, plants, distribution centers). Transportation lead times and capacities couple decisions across time and nodes. Feasibility and cost are highly sensitive to lead-time modeling: a one-period lead-time error can shift availability and trigger emergency purchases or stockouts.

\subsection{Contractual procurement constraints}
Contract terms that materially affect planning include MOQs, lead-time commitments, tiered pricing eligibility, allocation caps, order cadence requirements, and substitution approvals. Many such terms appear in tables, exhibits, or cross-referenced sections, which makes extraction error-prone. Contract2Plan assumes procurement text is authoritative for many terms, but unstructured, noisy, and sometimes contradictory.

\section{Task Definition: Contract-Grounded Planning}
\subsection{Inputs}
An instance contains:
\begin{itemize}
\item \textbf{Documents} $\mathcal{D}$: contracts, addenda, tenders, emails, catalogs (possibly OCR-derived).
\item \textbf{Master data} $\mathcal{M}$: supplier IDs, item master, AVL, unit/calendar normalization rules.
\item \textbf{BOM} $\mathcal{B}$: coefficients $a_{p,f}$ and substitute/alternate relationships.
\item \textbf{Network} $\mathcal{G}$: echelons/nodes, arcs, shipping lead times/costs/capacities.
\item \textbf{Demand} $d_{n,f,t}$ for finished goods over $t=1,\dots,T$.
\end{itemize}

\subsection{Outputs}
The system returns:
\begin{itemize}
\item a structured constraint set $\mathcal{C}$ with field-level evidence spans and confidence,
\item a feasible plan $\Pi$ (orders, production, shipments, inventories; optional backlog or emergency buys depending on the service model),
\item decision cards $E$ mapping decisions to binding constraints and provenance spans.
\end{itemize}

\subsection{Evaluation axes}
Contract2Plan is designed around four axes:
\begin{enumerate}
\item \textbf{Extraction quality:} field accuracy and evidence correctness (span-level grounding).
\item \textbf{Feasibility and compliance:} plans satisfy compiled constraints; conservative repair provides contract-safe feasibility for specific clause classes.
\item \textbf{Decision quality:} cost and service performance under the chosen execution/service model.
\item \textbf{Auditability:} each material constraint and decision is traceable to document evidence.
\end{enumerate}

\section{Method: Contract2Plan}
\subsection{System architecture}
Figure~\ref{fig:arch} shows the high-level pipeline. The design principle is explicit: extraction proposes constraints, but the solver verifier is the compliance gate. The system does not emit a plan unless constraints are grounded and the compiled model is feasible.

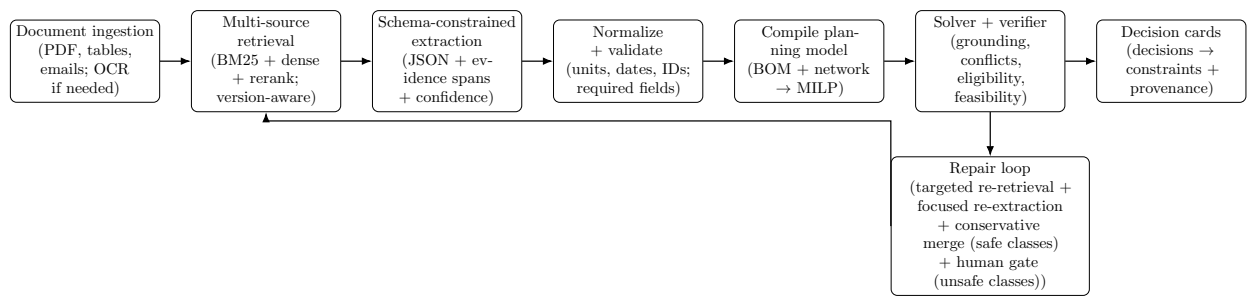
\begin{figure}[H]
\centering
\resizebox{\textwidth}{!}{%
\begin{tikzpicture}[
  font=\small,
  box/.style={
    draw, rounded corners, align=center,
    text width=3.1cm, minimum height=1.05cm
  },
  arrow/.style={-Latex, thick},
  node distance=8mm and 7mm
]
\node[box] (ingest) {Document ingestion\\(PDF, tables, emails; OCR if needed)};
\node[box, right=of ingest] (retrieve) {Multi-source retrieval\\(BM25 + dense + rerank;\\version-aware)};
\node[box, right=of retrieve] (extract) {Schema-constrained extraction\\(JSON + evidence spans\\+ confidence)};
\node[box, right=of extract] (norm) {Normalize + validate\\(units, dates, IDs;\\required fields)};
\node[box, right=of norm] (compile) {Compile planning model\\(BOM + network\\$\rightarrow$ MILP)};
\node[box, right=of compile] (verify) {Solver + verifier\\(grounding, conflicts,\\eligibility, feasibility)};
\node[box, right=of verify] (cards) {Decision cards\\(decisions $\rightarrow$\\constraints + provenance)};

\node[box, below=10mm of verify, text width=4.2cm] (repair) {Repair loop\\(targeted re-retrieval + focused re-extraction\\+ conservative merge (safe classes)\\+ human gate (unsafe classes))};

\draw[arrow] (ingest) -- (retrieve);
\draw[arrow] (retrieve) -- (extract);
\draw[arrow] (extract) -- (norm);
\draw[arrow] (norm) -- (compile);
\draw[arrow] (compile) -- (verify);
\draw[arrow] (verify) -- (cards);

\draw[arrow] (verify) -- (repair);
\draw[arrow] (repair.west) |- ($(retrieve.south)+(0,-2mm)$) -- (retrieve.south);

\end{tikzpicture}
}
\caption{Contract2Plan pipeline. Constraints are not trusted until they pass grounding, consistency, eligibility, and solver-based feasibility checks. Failures trigger targeted repair or explicit human confirmation.}
\label{fig:arch}
\end{figure}

\subsection{Constraint schema with provenance}
Contract2Plan uses a typed schema where each field includes: value, normalized units, applicability conditions, evidence spans, and confidence. A field is blocked if it lacks evidence spans. This makes the system robust to hallucinations: an ungrounded field cannot silently change the feasible region.

\begin{table}[H]
\centering
\begin{tabular}{llp{7.1cm}}
\toprule
Field & Type & Notes (what the verifier enforces) \\
\midrule
supplier\_id, part\_id & string & Must map to master data; ambiguous mapping triggers a gate. \\
effective\_start/end & date & Used for precedence and applicability windows. \\
moq & integer / null & Null only if evidence explicitly states ``no MOQ''. \\
lead\_time\_periods & int / range & Must be normalized to planning periods. \\
capacity\_per\_period & integer / null & Allocation caps; scope must be explicit if applicable. \\
price\_tiers & list of (threshold, unit\_price) & Eligibility constraints must be enforced in optimization. \\
substitution\_policy & enum & Allowed / forbidden / allowed-with-approval; evidence required. \\
evidence\_spans & list & (doc\_id, start, end, version) to support auditing. \\
confidence & float & Drives repair ordering and abstention. \\
\bottomrule
\end{tabular}
\caption{Planning-focused extraction schema. Every non-default field must carry evidence spans; otherwise it is gated.}
\label{tab:schema}
\end{table}

\subsection{Document ingestion and table handling}
Many failures originate before extraction:
\begin{itemize}
\item \textbf{Version drift:} addenda supersede base agreements, but are indexed without precedence metadata.
\item \textbf{Table parsing:} pricing schedules and allocation tables may be embedded as scanned images or complex layouts.
\item \textbf{Scope leakage:} chunking can remove headers that define applicability (dates, regions, SKUs).
\end{itemize}
Contract2Plan recommends clause-aware chunking (header + clause) and table-aware extraction for pricing and tier schedules. Each chunk is stored with document ID, version metadata, and offsets to enable verifiable provenance.

\subsection{Retrieval with provenance and versioning}
Retrieval is hybrid and version-aware:
\begin{itemize}
\item Sparse retrieval (BM25) for exact phrasing and identifiers.
\item Dense retrieval for paraphrases (``minimum purchase'' vs ``MOQ'').
\item Reranking to select top-$k$ evidence spans per field query.
\item Versioning and scope filters using document metadata (addendum vs master; effective dates; SKU scope).
\end{itemize}
All evidence spans are stored as provenance objects $(doc\_id, version, start, end)$ to support auditing and deterministic re-checks.

\subsection{Schema-constrained extraction and normalization}
Extraction emits JSON constrained to the schema. A deterministic normalization step then:
\begin{itemize}
\item converts units (days to periods; currency normalization when applicable),
\item parses tier tables and thresholds,
\item canonicalizes dates and effective windows,
\item links supplier/part strings to master IDs (with ambiguity handling).
\end{itemize}
Normalization is treated as a safety step: many production failures are unit or ID linkage errors rather than subtle semantic errors.

\subsection{Constraint consolidation and precedence resolution}
Conflicts are expected. The verifier clusters extracted fields by:
\[
(\text{supplier},\ \text{part},\ \text{field},\ \text{scope},\ \text{effective window})
\]
and attempts precedence resolution using:
\begin{enumerate}
\item effective dates (newer supersedes older),
\item explicit amendment language when available,
\item document-type ranking (signed addendum $>$ master agreement $>$ email).
\end{enumerate}
If precedence remains ambiguous, Contract2Plan applies conservative repair only for clause classes where conservatism is provably contract-safe (Section~\ref{sec:compliance}). Otherwise it abstains and requests confirmation.

\subsection{Optimization compilation (deterministic MILP backbone)}
Contract2Plan compiles consolidated constraints into a planning model. We present a deterministic MILP backbone sufficient for compliance verification and auditing. Robust or stochastic extensions are compatible but not required for the core compliance theorems.

The deterministic MILP backbone is intentionally minimal and compliance-oriented; richer stochastic, robust, or decomposition-based planning models can be substituted without changing the verifier logic, constraint grounding requirements, or abstention guarantees.

\textbf{Indices and sets.} Suppliers $s\in S$, parts $p\in P$, finished goods $f\in F\subseteq P$, nodes $n\in N$, time $t\in\{1,\dots,T\}$.

\textbf{Decision variables.} Orders $x_{s,p,t}\ge 0$ and activation $z_{s,p,t}\in\{0,1\}$; inventory $I_{n,p,t}\ge 0$; production $y_{n,f,t}\ge 0$; emergency buys $e_{n,p,t}\ge 0$ (a recourse variable used for feasibility under strict service).

\textbf{MOQ constraints.}
\begin{align}
x_{s,p,t} &\ge MOQ_{s,p}\, z_{s,p,t}, \label{eq:moq1}\\
x_{s,p,t} &\le M_{s,p}\, z_{s,p,t}. \label{eq:moq2}
\end{align}

\textbf{Allocation / capacity caps.}
\begin{equation}
x_{s,p,t} \le Cap_{s,p,t}\quad \forall s,p,t. \label{eq:cap}
\end{equation}

\textbf{Tier eligibility (compliance-critical).} Pricing is not a post-hoc annotation. The model must enforce that a discount tier is selectable only if the order meets its threshold. With tier binaries $u_{s,p,t,k}\in\{0,1\}$ and thresholds $\tau_k$:
\begin{align}
\sum_{k=1}^K u_{s,p,t,k} &= z_{s,p,t}, \label{eq:tierone}\\
x_{s,p,t} &\ge \sum_{k=1}^K \tau_k\, u_{s,p,t,k}. \label{eq:tierelig}
\end{align}

\subsection{Verifier and repair loop (expanded and auditable)}
The verifier is the compliance gate between extraction and execution. It is multi-layered so most failures are caught before solver calls, while solver diagnostics provide principled localization for hard infeasibility cases.

\textbf{Verifier responsibilities.} Given extracted constraints $\mathcal{C}$ and evidence spans $R$, the verifier must:
\begin{enumerate}
\item \textbf{Block unsupported fields:} no evidence, no effect on the optimization model.
\item \textbf{Enforce grounding integrity:} each numeric field must be supported by evidence spans (including table cells when applicable).
\item \textbf{Resolve contradictions safely:} precedence first; conservative merges only for provably safe classes; otherwise abstention.
\item \textbf{Enforce eligibility constraints:} tiers and substitutions must be encoded as constraints.
\item \textbf{Guarantee feasibility:} never emit a plan from an infeasible compiled model.
\item \textbf{Provide actionable diagnostics:} localize failures and suggest targeted repairs rather than generic retries.
\end{enumerate}

\textbf{Layer 1: schema and unit validation.} Checks include required identifiers, numeric sanity, tier-threshold monotonicity, canonical unit conversion, and well-formed applicability windows.

\textbf{Layer 2: provenance and grounding enforcement.} Each field must have evidence spans that support the value \emph{in context}. For example, if an evidence span contains ``MOQ 100'' but the scope header indicates it applies only to a different SKU family, the field is flagged as mis-scoped. Ungrounded or mis-scoped fields are blocked or sent to repair.

\textbf{Layer 3: cross-document consistency and applicability.} The verifier clusters constraints by $(supplier, part, field, effective\_window, scope)$, flags overlapping conflicts, and attempts precedence resolution. If unresolved, it applies conservative merges only for monotone feasibility constraints (Section~\ref{sec:compliance}); otherwise it gates to a human. The output of this layer is a consolidated constraint set where each consolidated value is either (i) resolved by precedence, or (ii) conservatively merged with an explicit justification.

\textbf{Layer 4: solver-based feasibility and infeasibility localization.} The verifier compiles the MILP and runs a feasibility check. If infeasible, it produces:
\begin{itemize}
\item IIS-based diagnosis when supported (a minimal conflicting subset), or
\item slack-minimization diagnosis otherwise:
\[
\min \sum_j w_j\xi_j \quad \text{s.t.}\quad g_j(x)\le \xi_j,\ \xi_j\ge 0.
\]
Large slacks identify which constraint families (MOQ, lead time, capacity, service targets) drive infeasibility.
\end{itemize}

\textbf{Repair actions.} Repairs are targeted: re-retrieve evidence for disputed fields, re-extract only those fields, normalize, and recompile. For unsafe clause types (exceptions, carve-outs, approvals), Contract2Plan abstains and returns minimal clarifying questions (for example, ``Which addendum is authoritative for MOQ on part P in region R after date D?'').

\subsection{Verifier/repair flow diagram}
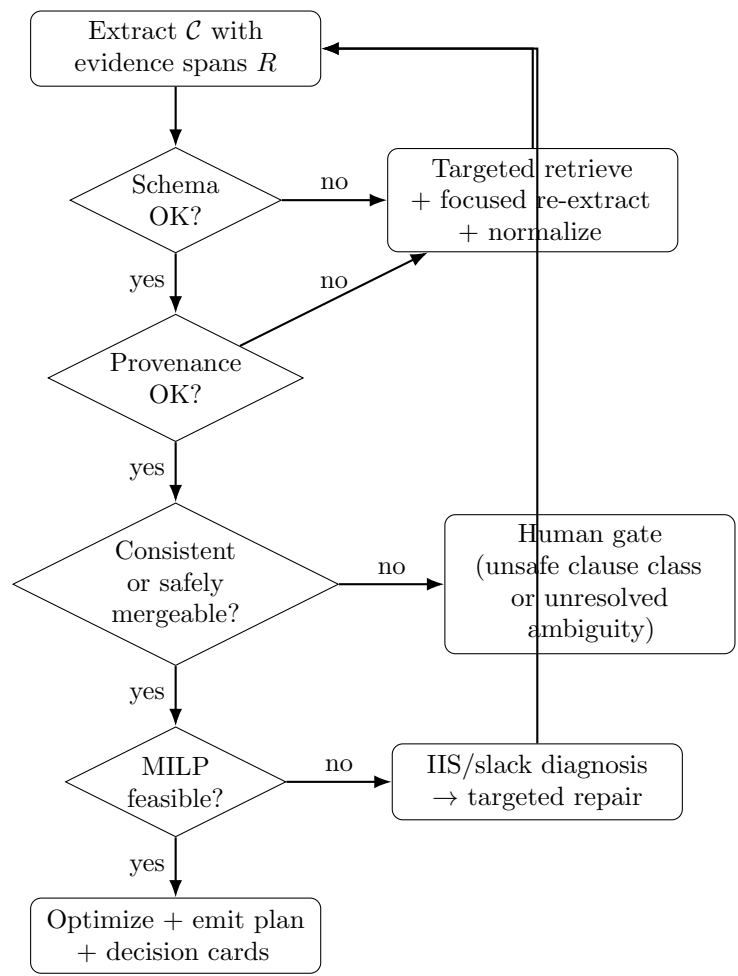
\begin{figure}[H]
\centering
\begin{tikzpicture}[
  font=\small,
  box/.style={draw, rounded corners, align=center, text width=3.6cm, minimum height=1.0cm},
  decision/.style={draw, shape=diamond, align=center, aspect=2, inner sep=1.5pt},
  arrow/.style={-Latex, thick},
  node distance=8mm and 10mm
]
\node[box] (c0) {Extract $\mathcal{C}$ with evidence spans $R$};
\node[decision, below=of c0] (schema) {Schema\\OK?};
\node[decision, below=of schema] (prov) {Provenance\\OK?};
\node[decision, below=of prov] (cons) {Consistent\\or safely\\mergeable?};
\node[decision, below=of cons] (feas) {MILP\\feasible?};

\node[box, right=14mm of schema] (repair1) {Targeted retrieve\\+ focused re-extract\\+ normalize};
\node[box, right=14mm of cons] (gate) {Human gate\\(unsafe clause class\\or unresolved ambiguity)};
\node[box, right=14mm of feas] (repair2) {IIS/slack diagnosis\\$\rightarrow$ targeted repair};

\node[box, below=of feas] (emit) {Optimize + emit plan\\+ decision cards};

\draw[arrow] (c0) -- (schema);
\draw[arrow] (schema) -- node[left]{yes} (prov);
\draw[arrow] (prov) -- node[left]{yes} (cons);
\draw[arrow] (cons) -- node[left]{yes} (feas);
\draw[arrow] (feas) -- node[left]{yes} (emit);

\draw[arrow] (schema) -- node[above]{no} (repair1);
\draw[arrow] (prov) -- node[above]{no} (repair1);
\draw[arrow] (repair1) |- (c0);

\draw[arrow] (cons) -- node[above]{no} (gate);
\draw[arrow] (feas) -- node[above]{no} (repair2);
\draw[arrow] (repair2) |- (c0);

\end{tikzpicture}
\caption{Verifier and repair logic. Contract2Plan loops on targeted evidence collection and constrained re-extraction. It escalates to humans when automation is unsafe.}
\label{fig:verifyflow}
\end{figure}

\subsection{Explainability via decision cards}
To support auditing, Contract2Plan emits a ``decision card'' per major recommendation (order, substitution choice, selected tier). A decision card includes:
\begin{itemize}
\item the decision (e.g., ``Order 200 units of part P from supplier S at period 3''),
\item the binding constraints that drove the decision (e.g., MOQ, tier threshold, capacity cap),
\item provenance spans for each binding constraint (document ID, offsets, version),
\item a short sensitivity note (which constraint change would alter the decision).
\end{itemize}
This makes it feasible for reviewers to confirm compliance and to spot systematic extraction errors (for example, if all MOQs trace to an outdated master agreement version).

\subsection{Contract excerpt walkthrough (document-like synthetic example)}
\label{sec:walkthrough}
This subsection provides a concrete, document-like walkthrough showing how Contract2Plan turns contract text into audited constraints and a solver-verified plan.
The excerpt is synthetic (for confidentiality) but written to reflect typical procurement language: scoped applicability, effective dates, tier schedules, and allocation caps.

\begin{figure}[H]
\centering
\fbox{%
\begin{minipage}{0.95\textwidth}
\small
\textbf{SUPPLY AGREEMENT ADDENDUM No. 3} \hfill \textbf{Effective: 2025-05-01}\\
\textbf{Supplier:} SUP-17 \hfill \textbf{Site scope:} MX-01\\[2pt]
\textbf{L1. Section 2.1 Minimum Order Quantity (MOQ).}
For Part \#88321 (``MCU-17'') delivered to MX-01, the MOQ per PO line is \textbf{150 units}.\\
\textbf{L2. Section 2.2 Volume condition.}
If cumulative quarterly volume for Part \#88321 to MX-01 is at least \textbf{600 units}, the MOQ is reduced to \textbf{100 units} for subsequent POs in that quarter.\\
\textbf{L3. Section 4 Lead Time.}
Standard lead time is \textbf{6 weeks}. Orders placed between \textbf{01-Aug} and \textbf{31-Oct} incur lead time \textbf{10 weeks}.\\
\textbf{L4. Exhibit B Price Schedule (per PO line).}
100--149 units: \$12.00 each; 150--299 units: \$11.30 each; $\ge$300 units: \$10.90 each.\\
\textbf{L5. Section 7 Allocation.}
Supplier may cap shipments to \textbf{250 units per month} for MX-01.\\
\end{minipage}}
\caption{Synthetic contract excerpt used for the walkthrough (written to resemble typical addendum structure).}
\label{fig:walkthrough_excerpt}
\end{figure}

\paragraph{Retrieval and extraction.}
Given a part master query (part name, part ID, supplier ID, site), the retriever pulls the spans containing MOQ, lead time, tiers, and allocation language.
The extractor emits a typed schema with evidence pointers back to specific lines in Figure~\ref{fig:walkthrough_excerpt}.

\begin{figure}[H]
\centering
\fbox{%
\begin{minipage}{0.95\textwidth}
\small\ttfamily
\{\\
\ \ "doc\_id": "Addendum-3",\\
\ \ "supplier\_id": "SUP-17",\\
\ \ "part\_id": "88321",\\
\ \ "scope": \{ "site": "MX-01" \},\\
\ \ "effective\_start": "2025-05-01",\\
\ \ "moq": 150,\\
\ \ "lead\_time\_weeks": \{ "standard": 6, "peak\_season": 10, "peak\_window": "Aug--Oct" \},\\
\ \ "capacity\_per\_month": 250,\\
\ \ "price\_tiers": [\\
\ \ \ \ \{ "threshold": 100, "unit\_price": 12.00 \},\\
\ \ \ \ \{ "threshold": 150, "unit\_price": 11.30 \},\\
\ \ \ \ \{ "threshold": 300, "unit\_price": 10.90 \}\\
\ \ ],\\
\ \ "conditions": [\\
\ \ \ \ \{ "type": "volume", "threshold": 600, "effect": "moq=100 for subsequent POs in-quarter" \}\\
\ \ ],\\
\ \ "evidence": [ "Addendum-3:L1", "Addendum-3:L2", "Addendum-3:L3", "Addendum-3:L4", "Addendum-3:L5" ]\\
\}\\
\end{minipage}}
\caption{Example schema-constrained extraction output with evidence pointers to the excerpt lines.}
\label{fig:walkthrough_json}
\end{figure}

\paragraph{Normalization and compilation.}
The normalization step converts lead time to planning periods (given the calendar definition), canonicalizes the site scope, and validates tier monotonicity.
Compilation then creates the MILP constraints for MOQ (Eqs.~(1)--(2)), tier eligibility (Eqs.~(4)--(5)), and allocation (Eq.~(3)).

\paragraph{Verifier behavior on conditional clauses.}
Line L2 introduces a conditional MOQ reduction tied to cumulative quarterly volume.
If the system has the necessary aggregation model and time alignment to encode this condition safely, it may compile it explicitly.
If not, Contract2Plan treats the conditional as a non-trivial applicability constraint: it either (i) conservatively enforces MOQ=150 (contract-safe for feasibility, potentially cost-suboptimal), or (ii) gates to human confirmation if the condition materially changes feasibility or economics.
In both cases, the decision card explicitly cites the evidence line and states whether the conditional was encoded or conservatively collapsed.

\paragraph{Outcome: auditable decisions.}
The emitted plan is accompanied by decision cards that identify which constraints were binding (MOQ, allocation, lead time window) and provide provenance back to Figure~\ref{fig:walkthrough_excerpt}.
This walkthrough illustrates the central design: extraction proposes constraints, but the verifier controls whether and how they affect the feasible region.

\section{Compliance Theorems and Abstention Policy}
\label{sec:compliance}
This section formalizes which constraint classes can be handled by conservative repair (safely) and which must be gated to human confirmation. The goal is not to over-claim legal correctness, but to precisely delimit what the system can guarantee under explicit assumptions.

\subsection{Constraint taxonomy}
We partition extracted terms into three practical classes.

\textbf{Class A: monotone feasibility constraints (safe for conservative repair).}
These shrink the feasible set in a known direction:
\begin{itemize}
\item MOQ: larger is more restrictive.
\item Lead time: larger (later arrival) is more restrictive.
\item Capacity caps: smaller is more restrictive.
\item Minimum order intervals: larger is more restrictive.
\end{itemize}

\textbf{Class B: eligibility constraints (safe if encoded as constraints).}
Economic terms like pricing are not monotone, but eligibility often is. For price breaks, it is safe to enforce that a tier can only be applied if its threshold is met (Equations \eqref{eq:tierone}--\eqref{eq:tierelig}). This guarantees that the solver cannot claim an ineligible tier.

\textbf{Class C: non-monotone or exception-heavy clauses (unsafe for automatic merge).}
These include nested exceptions, carve-outs, ambiguous cross-references, or approvals tied to external processes. Conservative numeric merging is not semantically safe here. Contract2Plan abstains unless authoritative precedence metadata resolves ambiguity.

\begin{definition}[Restrictiveness order]
For a constraint type $\tau$ in Class A, define a partial order $\preceq_\tau$ such that
$c_1 \preceq_\tau c_2$ means $c_2$ is at least as restrictive as $c_1$ (the feasible set under $c_2$ is a subset of that under $c_1$ for type $\tau$).
\end{definition}

\subsection{Conservative merge operator}
Let $\mathcal{V}_\tau$ be candidate values for a Class A constraint type $\tau$ for a fixed $(supplier, part, scope, effective\_window)$ when precedence cannot be established. Define:
\[
\mathrm{merge}_\tau(\mathcal{V}_\tau)=
\begin{cases}
\max \mathcal{V}_\tau, & \text{if larger is more restrictive (MOQ, lead time, interval),}\\
\min \mathcal{V}_\tau, & \text{if smaller is more restrictive (capacity cap).}
\end{cases}
\]

\begin{theorem}[Conservative feasibility implies contract-safe feasibility for Class A]
\label{thm:conservative}
Fix a planning instance (BOM, network, demand). For each Class A constraint type $\tau$, let $\mathcal{V}_\tau$ be the retrieved candidate values applicable to the same scope and effective window, and let $\widehat{c}_\tau=\mathrm{merge}_\tau(\mathcal{V}_\tau)$. Let $\widehat{\mathcal{C}}$ denote the conservatively merged constraint set.

Assume:
\begin{enumerate}
\item \textbf{Coverage:} the true applicable value $c^\star_\tau$ is contained in $\mathcal{V}_\tau$ for each $\tau$ (retrieval did not miss the strictest applicable clause).
\item \textbf{Correct monotone encoding:} the MILP encoding is monotone with respect to $\preceq_\tau$.
\end{enumerate}
Then any plan feasible under $\widehat{\mathcal{C}}$ is feasible under the true constraints $\mathcal{C}^\star$ for all Class A constraint types.
\end{theorem}

\begin{proof}[Proof sketch]
For each $\tau$, $\widehat{c}_\tau$ is at least as restrictive as every element in $\mathcal{V}_\tau$, hence at least as restrictive as $c^\star_\tau$ (Coverage). By monotonicity of the encoding, the feasible set under $\widehat{c}_\tau$ is a subset of the feasible set under $c^\star_\tau$. Intersecting over $\tau$ yields the claim.
\end{proof}

\begin{proposition}[Tier eligibility prevents ineligible discount claims]
If tier $k$ is selectable only when the order quantity satisfies $x_{s,p,t}\ge \tau_k$ (Equation \eqref{eq:tierelig}), then no feasible solution can apply tier $k$ without meeting its threshold. Thus the optimizer cannot claim an ineligible price tier.
\end{proposition}

\subsection{Abstention policy: when humans must confirm}
Contract2Plan abstains (requests confirmation) when:
\begin{itemize}
\item a field lacks evidence spans (grounding failure),
\item conflicts remain in Class C clauses (exceptions, carve-outs, ambiguous cross-references),
\item scope/effective windows cannot be resolved and a conservative merge would change meaning,
\item substitutions require approvals not present in master data/AVL.
\end{itemize}
In practice, abstention frequency is deployment- and document-dependent, varying with contract heterogeneity and governance maturity; in environments with well-versioned addenda and clearly scoped clauses, most abstentions arise from genuinely non-monotone exception logic rather than routine MOQ, lead-time, or capacity terms.

\begin{table}[H]
\centering
\begin{tabular}{p{3.4cm}p{5.6cm}p{5.1cm}}
\toprule
Clause type & Safe automated action & Gate when ambiguous \\
\midrule
MOQ, lead time, capacity, cadence (Class A) & Precedence resolve; else conservative merge + verify & Human confirmation if Coverage likely missing \\
Price breaks / tiers (Class B) & Enforce eligibility; do not understate thresholds & Human confirmation if tier table OCR corrupt/unclear \\
Penalty clauses with exceptions (Class C) & Do not auto-merge; treat as policy parameters only when explicit & Human/legal review \\
Substitution approvals / AVL & Forbid substitution unless approval evidence exists & Human quality/compliance approval \\
Delivery terms affecting responsibility & Treat as scope metadata; avoid guessing costs & Human confirmation if responsibility unclear \\
\bottomrule
\end{tabular}
\caption{Constraint taxonomy and abstention policy. Contract2Plan is conservative where guarantees are justified and abstains where conservative merging is semantically unsafe.}
\label{tab:taxonomy}
\end{table}

\section{Evaluation}
\paragraph{Evaluation scope.}
Our evaluation is intentionally \emph{risk-illustrative} rather than performance-exhaustive:
we isolate key error modes (MOQ and lead-time mis-extraction) under a transparent execution model to quantify feasibility and tail-risk impacts.
We do not claim that the micro-benchmark captures the full diversity of industrial contract language, network structures, or operational policies.

We provide (i) a checkable toy example with explicit numeric calculations, (ii) a worked example illustrating BOM coupling and substitution gating, and (iii) a computed synthetic micro-benchmark quantifying the economic and compliance risk of MOQ and lead-time extraction errors.

\subsection{Toy example (fully computed by hand)}
Consider a single item over $T=3$ periods with demands $(50,50,50)$, initial inventory $0$, holding cost $h=0.1$ per unit per period, cheap supplier unit cost $c_{\text{cheap}}=10$, emergency cost $c_{\text{exp}}=20$. The true contract terms are MOQ $=100$ and lead time $L_{\text{true}}=2$ periods.

A feasible cheap order must be either $0$ or at least $100$, and if ordered in period 1 it arrives in period 3. Under the strict no-backlog service model, demand in periods 1 and 2 must be met via emergency purchases. If we order 100 in period 1, then execution cost is:
\[
\text{cheap} = 100\cdot 10 = 1000,\quad
\text{emergency} = (50+50)\cdot 20 = 2000,\quad
\text{holding} = \text{end inventory in period 3} \cdot 0.1.
\]
In period 3, the 100 units arrive and 50 are consumed, leaving 50 units; holding cost adds $50\cdot 0.1=5$. Total cost is $1000+2000+5=3005$.

Now suppose the extractor underestimates lead time as $L_{\text{ext}}=1$ and still outputs MOQ $=100$. The planner might choose to order 100 in period 2 (believing it arrives in period 3) instead of period 1, expecting to reduce holding. In execution, the order arrives in period 4 (outside the horizon), so all three periods are covered by emergency purchases: $150\cdot 20=3000$. Total executed cost becomes $3000+100\cdot 10=4000$, yielding regret $995$. This illustrates how a one-period lead-time underestimation can dominate cost even in a minimal setting, motivating solver-based verification and conservative handling of lead-time clauses.

\subsection{Worked example: BOM with an alternate and approval gate}
Figure~\ref{fig:bom} shows a two-level BOM. Finished good $FG$ requires subassembly $A$ and component $B$. Subassembly $A$ requires component $C$. Component $B'$ is an alternate for $B$ allowed only with explicit approval evidence.

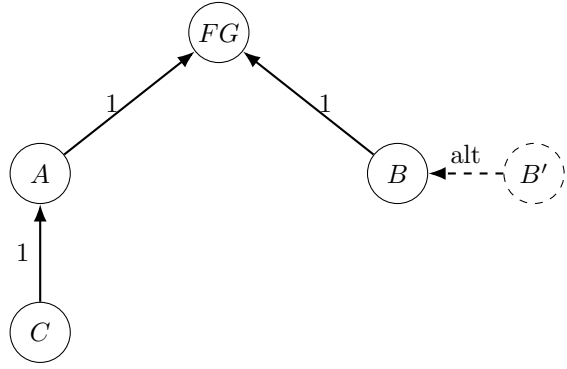
\begin{figure}[H]
\centering
\begin{tikzpicture}[
  font=\small,
  node/.style={draw, circle, minimum size=8mm, inner sep=0pt},
  alt/.style={draw, circle, dashed, minimum size=8mm, inner sep=0pt},
  arrow/.style={-Latex, thick}
]
\node[node] (FG) {$FG$};
\node[node, below left=13mm and 18mm of FG] (A) {$A$};
\node[node, below right=13mm and 18mm of FG] (B) {$B$};
\node[alt, right=10mm of B] (Bp) {$B'$};
\node[node, below=13mm of A] (C) {$C$};

\draw[arrow] (A) -- node[left]{1} (FG);
\draw[arrow] (B) -- node[right]{1} (FG);
\draw[arrow] (C) -- node[left]{1} (A);

\draw[arrow, dashed] (Bp) -- node[above]{alt} (B);

\node[draw, rounded corners, align=left, text width=10.6cm, below=14mm of C] (note) {
\textbf{Substitution rule:} $B'$ can replace $B$ only if the AVL or contract provides approval evidence. Contract2Plan enforces this as a constraint: if approval is missing, $B'$ is forbidden and the system escalates rather than guessing.};

\end{tikzpicture}
\caption{Example BOM with an alternate component. Substitution is compliance-sensitive and is enforced as a constraint.}
\label{fig:bom}
\end{figure}

\subsection{Synthetic micro-benchmark (computed by exact enumeration)}
Large public datasets combining full contract text, BOMs, and multi-echelon networks are rare. Instead of fabricating empirical claims, we provide a focused micro-benchmark that is fully specified and computed exactly.

\textbf{Benchmark setting.} Single-item replenishment over horizon $T=5$. Demand must be met each period (no backlog). Two purchasing modes:
\begin{itemize}
\item a cheap supplier with MOQ and lead time, unit cost $c_{\text{cheap}}$,
\item an emergency purchase with immediate availability, unit cost $c_{\text{exp}}>c_{\text{cheap}}$.
\end{itemize}
Execution model: if the plan orders $0<q<MOQ_{\text{true}}$, the supplier uplifts to $MOQ_{\text{true}}$ (and the buyer pays for the uplifted quantity). If arrivals are late due to lead-time underestimation, emergency purchases fill the shortfall.

\textbf{Exact enumeration.} Cheap orders each period are restricted to:
\[
\{0,50,100,150,200,300,400,450,600\},
\]
so there are $9^5=59{,}049$ schedules per instance. For each instance we enumerate all schedules and select the minimum-cost one under extracted constraints, then execute it under true constraints.

\textbf{Instance generator and extraction error model.} We generate 500 instances with random seed 42:
\begin{itemize}
\item $d_t \sim \mathrm{Unif}\{0,\dots,80\}$ i.i.d.,
\item $L_{\text{true}} \sim \mathrm{Unif}\{1,2,3\}$,
\item $MOQ_{\text{true}} \sim \mathrm{Unif}\{50,100,150,200\}$,
\item $c_{\text{cheap}} \sim \mathrm{Unif}[6,12]$,
\item $c_{\text{exp}} = c_{\text{cheap}} + \mathrm{Unif}[4,12]$,
\item holding cost $h \sim \mathrm{Unif}[0.02,0.2]$.
\end{itemize}
Extraction errors: MOQ under-step by one level with probability 0.30; over-step by one with probability 0.10; lead time underestimated by 1 with probability 0.25; overestimated by 1 with probability 0.10 (clipped to the valid set).

\textbf{Regret.} Let $\Pi_{\text{ext}}$ be the schedule that is optimal under extracted constraints but executed under true constraints (MOQ uplift and true lead time). Let $\Pi^\star$ be the optimal schedule under true constraints. Regret is executed cost minus true-optimal cost.

\begin{table}[H]
\centering
\begin{tabular}{lr}
\toprule
Metric (500 instances; exact enumeration; seed 42) & Value \\
\midrule
Instances with any MOQ violation in planned orders & 83 / 500 = 16.6\% \\
Mean optimal cost under true constraints & \$2,636.13 \\
Mean executed cost of extraction-only plan & \$2,778.46 \\
Mean regret (absolute) & \$142.33 \\
Mean regret / mean optimal cost & 5.40\% \\
Median regret & \$0.00 \\
90th percentile regret & \$587.74 \\
95th percentile regret & \$955.83 \\
99th percentile regret & \$1,569.61 \\
Maximum regret & \$2,242.22 \\
Fraction with regret $>0$ & 27.2\% \\
Mean regret conditional on MOQ violation & \$523.68 \\
90th percentile regret conditional on MOQ violation & \$1,418.43 \\
\bottomrule
\end{tabular}
\caption{Micro-benchmark results computed by exact enumeration. Extraction errors induce heavy-tailed economic risk and nontrivial compliance risk.}
\label{tab:micro1}
\end{table}

\begin{table}[H]
\centering
\begin{tabular}{llrrrr}
\toprule
MOQ relation & Lead time relation & \# inst & Mean regret & Median regret & MOQ viol. rate \\
\midrule
equal & equal & 265 & 0.00 & 0.00 & 0.0\% \\
equal & over & 21 & 160.93 & 42.68 & 0.0\% \\
equal & under & 70 & 242.74 & 12.25 & 0.0\% \\
over & equal & 21 & 102.53 & 9.71 & 0.0\% \\
over & over & 4 & 699.87 & 758.18 & 0.0\% \\
over & under & 7 & 280.04 & 209.44 & 0.0\% \\
under & equal & 83 & 343.02 & 0.00 & 71.1\% \\
under & over & 5 & 467.75 & 0.00 & 80.0\% \\
under & under & 24 & 544.70 & 214.57 & 83.3\% \\
\bottomrule
\end{tabular}
\caption{Error-pattern decomposition. Joint underestimation of MOQ and lead time is the highest-risk regime, combining uplift-induced overbuying and late-arrival emergency purchases.}
\label{tab:micro2}
\end{table}

\section{Discussion}
\textbf{When Contract2Plan helps.} The approach is most valuable when (i) operational terms frequently change via addenda, (ii) pricing eligibility and MOQs materially affect feasibility and cost, (iii) BOM coupling amplifies upstream errors, and (iv) organizations require auditable decision trails.

\textbf{When it can fail.} Conservative repair can become overly restrictive and inflate cost, especially when the strictest retrieved constraint is not actually applicable (scope misunderstanding). More importantly, Coverage can fail: if retrieval misses a strict addendum, no conservative merge can guarantee compliance. Non-monotone clauses (rebates, penalty exceptions, legal carve-outs) require abstention and human confirmation.

\textbf{Practical deployment.} A realistic deployment should (i) use document governance (version metadata, effective dates, signatures), (ii) adopt role-based access control and audit logs, and (iii) treat the verifier as a hard gate before execution.

\section{Security, Governance, and Ethics}
Procurement documents are untrusted inputs: they may contain inconsistent terms, malformed tables, or adversarial instructions (prompt injection). OWASP identifies prompt injection and insecure output handling as major risks for LLM applications \cite{owasp_llm_top10}. Contract2Plan mitigates these risks by treating retrieved text as data, enforcing schema-only extraction, requiring evidence spans for all constraints, and preventing untrusted text from directly triggering tool actions.

Over-automation risk is real: procurement decisions can affect supplier fairness and compliance. Contract2Plan is designed as decision support with explicit abstention rules, not as an autonomous procurement agent. In deployment, procurement and legal teams should control the ``human gate'' policy and review workflows.

\section{Limitations}
Key limitations include:
\begin{itemize}
\item The conservative compliance theorem requires \emph{Coverage}; if retrieval misses the strictest applicable clause, conservative merging cannot guarantee compliance.
\item Many economically important clauses are non-monotone (rebates, exceptions); abstention is necessary but increases human workload.
\item The micro-benchmark is designed to highlight tail risk from specific extraction errors and should not be interpreted as an end-to-end performance leaderboard.
\item Industrial-scale BOMs and networks can yield large MILPs; decomposition and rolling-horizon strategies are needed in practice.
\item Our micro-benchmark is synthetic and isolates only MOQ and lead time errors; real deployments must evaluate broader clause families and document modalities (scanned exhibits, multi-currency pricing tables).
\end{itemize}

\section{Conclusion}
Contract2Plan operationalizes a key principle for GenAI-enabled planning: verification before action. By grounding contract-derived constraints in evidence spans, compiling them into an explicit planning model, and using solver-driven feasibility checks with targeted repair and conservative merges for provably safe clause classes, Contract2Plan produces auditable plans that are feasible and compliance-aware. A computed synthetic benchmark quantifies heavy-tailed economic and compliance risk from MOQ and lead-time extraction errors, motivating verification as a first-class component of contract-grounded planning systems.

\FloatBarrier
\clearpage
\appendix

\section{Verifier Pseudocode }
\begin{algorithm}[H]
\caption{Contract2Plan verifier and repair loop}
\label{alg:verifier}
\begin{algorithmic}[1]
\Require Documents $\mathcal{D}$, master data $\mathcal{M}$, BOM $\mathcal{B}$, network $\mathcal{G}$, demand $d$
\Ensure Feasible plan $\Pi$, decision cards $E$, audited constraints $\mathcal{C}$
\State $R \gets$ \textsc{RetrieveEvidence}$(\mathcal{D}, \mathcal{M})$
\State $\mathcal{C} \gets$ \textsc{ExtractConstraints}$(R)$
\State $\mathcal{C} \gets$ \textsc{Normalize}$(\mathcal{C}, \mathcal{M})$
\For{$iter = 1$ to $I_{\max}$}
  \State $(ok_s, issues_s) \gets$ \textsc{SchemaValidate}$(\mathcal{C})$
  \State $(ok_p, issues_p) \gets$ \textsc{ProvenanceCheck}$(\mathcal{C}, R)$
  \State $(ok_c, conflicts) \gets$ \textsc{ConsistencyCheck}$(\mathcal{C})$
  \If{not $ok_s$ or not $ok_p$}
     \State $R \gets$ \textsc{TargetedRetrieve}$(\mathcal{D}, issues_s, issues_p)$
     \State $\mathcal{C} \gets$ \textsc{FocusedReExtract}$(R, issues_s, issues_p)$
     \State $\mathcal{C} \gets$ \textsc{Normalize}$(\mathcal{C}, \mathcal{M})$
     \State \textbf{continue}
  \EndIf
  \If{not $ok_c$}
     \State $\mathcal{C} \gets$ \textsc{ResolveByPrecedence}$(\mathcal{C}, conflicts, \mathcal{M})$
     \State $(ok_c2, conflicts2) \gets$ \textsc{ConsistencyCheck}$(\mathcal{C})$
     \If{not $ok_c2$}
        \State $\mathcal{C} \gets$ \textsc{ConservativeMergeClassA}$(\mathcal{C}, conflicts2)$
        \State $(ok_c3, conflicts3) \gets$ \textsc{ConsistencyCheck}$(\mathcal{C})$
        \If{not $ok_c3$}
           \State \Return \textsc{HumanGate}$(conflicts3)$
        \EndIf
     \EndIf
  \EndIf
  \State $M \gets$ \textsc{CompileMILP}$(\mathcal{C}, \mathcal{M}, \mathcal{B}, \mathcal{G}, d)$
  \State $(feas, report) \gets$ \textsc{FeasibilityCheck}$(M)$
  \If{$feas$}
     \State $\Pi \gets$ \textsc{Optimize}$(M)$
     \State $E \gets$ \textsc{BuildDecisionCards}$(\Pi, \mathcal{C}, R)$
     \State \Return $(\Pi, E, \mathcal{C})$
  \Else
     \State $action \gets$ \textsc{SelectRepairAction}$(report)$
     \If{$action ==$ \textsc{TargetedRetrieve}}
        \State $R \gets$ \textsc{TargetedRetrieve}$(\mathcal{D}, report)$
        \State $\mathcal{C} \gets$ \textsc{FocusedReExtract}$(R, report)$
        \State $\mathcal{C} \gets$ \textsc{Normalize}$(\mathcal{C}, \mathcal{M})$
     \ElsIf{$action ==$ \textsc{ConservativeMerge}}
        \State $\mathcal{C} \gets$ \textsc{ConservativeMergeClassA}$(\mathcal{C}, report)$
     \Else
        \State \Return \textsc{HumanGate}$(report)$
     \EndIf
  \EndIf
\EndFor
\State \Return \textsc{HumanGate}$(\mathcal{C})$
\end{algorithmic}
\end{algorithm}

\section{Micro-benchmark reproducibility (exact enumeration)}
Horizon $T=5$ and action set $\{0,50,100,150,200,300,400,450,600\}$ yields $9^5=59{,}049$ schedules. For each instance:
\begin{enumerate}
\item enumerate all schedules and compute planning cost under extracted $(MOQ_{\text{ext}},L_{\text{ext}})$ (invalid if any $0<q<MOQ_{\text{ext}}$),
\item select the minimum-cost extracted-optimal schedule,
\item execute it under true $(MOQ_{\text{true}},L_{\text{true}})$ with MOQ uplift and emergency purchases for late arrivals,
\item compute regret relative to the true-optimal schedule (also computed by enumeration).
\end{enumerate}

\section{Reproducibility checklist (practical)}
A minimal reproducibility bundle should include:
\begin{itemize}
\item exact prompt templates for retrieval and schema extraction,
\item the schema definition (Table~\ref{tab:schema}) and normalization rules,
\item the precedence rules and abstention policy (Section~\ref{sec:compliance}),
\item MILP compiler settings (solver, tolerances, integrality gap, time limits),
\item verifier configuration (IIS vs slack mode, repair iteration limit),
\item code for the micro-benchmark generator and enumerator.
\end{itemize}

\FloatBarrier

\end{document}